\documentclass[11pt,a4paper,oneside,headings=small,abstracton,DIV=12]{scrartcl}
\usepackage{amsmath,amssymb,amsfonts,amsthm,enumerate}
\usepackage[english]{babel}
\usepackage{microtype,color}
\usepackage[blocks]{authblk}

\newtheorem{theorem}{Theorem}

\newtheorem{proposition}[theorem]{Proposition}

\theoremstyle{definition}
\newtheorem{definition}[theorem]{Definition}

\theoremstyle{remark}
\newtheorem{remark}[theorem]{Remark}
\newcommand{\RR}{\mathbb{R}}
\newcommand{\drm}{\mathrm{d}}
\newcommand{\euler}{\mathrm{e}}
\newcommand{\CC}{\mathbb{C}}
\newcommand{\NN}{\mathbb{N}}	
\newcommand{\ZZ}{\mathbb{Z}}
\newcommand{\EE}{\mathbb{E}}
\newcommand{\PP}{\mathbb{P}}
\newcommand{\Tr}{\operatorname{Tr}}

\newcommand{\cov}{\operatorname{Cov}}
\newcommand{\var}{\operatorname{Var}}
\newcommand{\Eins}{\mathbf{1}}

\newcommand{\cB}{\mathcal{B}}
\newcommand{\cA}{\mathcal{A}}
\DeclareMathOperator*{\esssup}{ess\,sup}
\usepackage[]{hyperref}
	\usepackage{color}
 	\definecolor{darkred}{rgb}{0.5,0,0}
 	\definecolor{darkgreen}{rgb}{0,0.5,0}
 	\definecolor{darkblue}{rgb}{0,0,0.5}  	\hypersetup{colorlinks,linkcolor=darkblue,filecolor=darkgreen,urlcolor=darkred,citecolor=darkblue}
\title{\Large Wegner estimate for discrete Schr\"odinger operators with Gaussian random potentials}
\author{Martin Tautenhahn}
\affil{Fakult\"at f\"ur Mathematik, Technische Universit\"at Chemnitz, Germany}
\date{\vspace{-2.5em}}
\begin{document}
\maketitle
\begin{abstract}
 We prove a Wegner estimate for discrete Schr\"odinger operators with a potential given by a Gaussian random process. The only assumption is that the covariance function decays exponentially, no monotonicity assumption is required. This improves earlier results where abstract conditions on the conditional distribution, compactly supported and non-negative, or compactly supported covariance functions with positive mean are considered.
\end{abstract}
\section{Introduction}
In this note we study the family of finite volume random Schr\"odinger operators 
\[
 H_{\omega , L} = A_L + \lambda V_{\omega , L}, \quad \omega \in \Omega,
\]
in $\ell^2 (\Lambda_L)$, where $\Lambda_L = [-L/2 , L/2]^d \cap \ZZ^d$, and $\lambda > 0$. Here, $A_L$ is an arbitrary self-adjoint operator, $(\Omega , \mathcal{A} , \PP)$ is a probability space, and the potential values $V_{\omega , L} (x)$, $x \in \Lambda_L$, are given by a non-degenerated stationary Gaussian process $V = \{V_x \colon \Omega \to \RR, \ x \in \ZZ^d\}$ with mean zero and covariance function $\gamma : \ZZ^d \to \RR$ satisfying $\lvert \gamma (x) \rvert \leq D \euler^{-\alpha \lvert x \rvert_1}$ for some positive constants $D$ and $\alpha$. Our main result is a Wegner estimate, that is, an upper bound of the expected number of eigenvalues in a bounded energy interval $I \subset \RR$ of the form 
\[
 \EE \left( \Tr \bigl(\chi_I (H_{\omega , L}) \bigr) \right) \leq
\frac{C_{\mathrm{W}}}{\lambda} \lvert I \rvert \lvert \Lambda_L \rvert^m .
\]
Here, $C_{\mathrm W} > 0$ and $m \geq 1$ are constants depending only on the model parameters. Let us emphasize that the only assumption is that the modulus of the covariance function decays exponentially. In particular, this allows long-range as well as non-monotone correlations at the same time. No monotonicity assumption is required. 
The present paper is inspired by an earlier joint project with Ivan Veseli\'c on Wegner estimates for random Schr\"odinger operators with Gaussian potentials in the continuum space $\RR^d$, whose results will be published in a companion paper.
\par
Wegner estimates serve as ingredients for proofs of localization via multiscale analysis. Here, localization refers to the phenomenon that parts of the spectrum (of the infinite volume operator in $\ell^2 (\ZZ^d)$) consist almost surely only of pure point spectrum (spectral localization), or that the solutions of the Schr\"odinger equation stay trapped in a finite region of space for all time (dynamical localization). The multiscale analysis is an induction argument over the scale $L$. The Wegner estimate establishes the induction step, while the induction anchor is provided by the so-called initial scale estimate. 
Let us note that our Wegner estimate allows to prove localization at all energies where an initial length scale estimate is provided. In particular, the initial length scale estimate follows from our Wegner estimate in the case of sufficiently large disorder $\lambda > 0$, see \cite{DreifusK-89,Kirsch-08}.
For the method of multiscale analysis we refer to the seminal papers \cite{FroehlichS-83,FroehlichMSS-85}, and to \cite{DreifusK-89,GerminetK-01,GerminetK-03,GerminetK-06}. 
If the covariance function does not have compact support, one has to use an enhanced version of the multiscale analysis to prove localization, see \cite{DreifusK-91} for the discrete, and \cite{KirschSS-98} for the continuum setting.
\par
While proofs of Wegner estimates and localization for random Schr\"odinger operators have initially been developed in the case where the potential values are independent and identically distributed, see \cite{FroehlichS-83,FroehlichMSS-85,DreifusK-89} for multiscale analysis, and \cite{AizenmanM-93,Aizenman-94,Graf-94} for the so-called fractional moment method, both methods have been subsequently extended to models where the potential values at different lattice sites are correlated random variables. We focus our discussion here on Gaussian random potentials. 
\par
In \cite{DreifusK-91,AizenmanM-93,AizenmanG-98,Hundertmark-00,AizenmanSFH-01,Hundertmark-08} random operators with so-called conditional $\tau$-H\"older continuous potential values are considered. This means, that the distribution of $V_0$ conditioned on fixed values $V_k$, $k \not = 0$, is (uniformly) $\tau$-H\"older continuous. In \cite{DreifusK-91} the authors construct an example of a Gaussian process which satisfies this assumptions. However, in Section~\ref{sec:conditional} we show that Gaussian processes are in general not conditional $\tau$-H\"older continuous, even not for compactly supported covariance functions.
 \par
In \cite{FischerHLM-97,Ueki-04,Veselic-11} random Schr\"odinger operators in $L^2 (\RR^d)$ with a Gaussian random potential are studied. Although these papers consider operators in the continuum, it is possible to transfer their results to the discrete setting. The paper \cite{Ueki-04} considers Schr\"odinger operators with a bounded vector potential and a Gaussian random scalar potential. The covariance function is assumed to have compact support and to be sufficiently regular. The paper \cite{FischerHLM-97} provides an abstract condition on the covariance function which is sufficient to obtain a Wegner estimate. This condition is satisfied if the covariance function is non-negative. Only one example of a sign-changing covariance function satisfying this abstract condition is constructed. 
Indeed, the paper \cite{FischerHLM-97} leaves it open whether this condition is applicable to a certain class of sign-changing covariance function or not. In Section~\ref{sec:abstract} we formulate a different condition on the covariance function which implies a Wegner estimate as well. Then we show how a certain tiling theorem from \cite{LeonhardtPTV-15} can be efficiently used to show that this new condition is satisfied for all sign-changing and exponentially decaying covariance functions. The paper \cite{Veselic-11} shows that the abstract condition from \cite{FischerHLM-97} is satisfied if the covariance function has compact support and positive mean. To the knowledge of the author, this result is most general for Gaussian models with a sign-changing covariance function. 
Beside the compact support, the only situation which cannot be treated with the methods from \cite{Veselic-11} is if the covariance function has mean zero. 
That this case is particularly difficult has been observed before in terms of the alloy-type model, see, e.g., \cite{Veselic-10b}.
\par
We summarize, that our Wegner estimate is applicable in situations, whereas non of the above mentioned papers applies.
\par
Non-monotone and long-range correlations have also been modeled and studied in terms of the so-called discrete and continuous alloy-type model, see, e.g.,  \cite{Klopp-95a,HislopK-02,Veselic-02,KostrykinV-06,Veselic-10,Veselic-10b,ElgartTV-10,ElgartTV-11,Krueger-12,ElgartSS-14,LeonhardtPTV-15}. Let us stress, that the papers \cite{Veselic-02,KostrykinV-06,Veselic-10,Veselic-10b,LeonhardtPTV-15} 
use a transformation of the probability space to obtain a Wegner estimate for alloy-type models under a certain condition on the so-called single-site potential. In particular, the condition on the single-site potential in \cite{KostrykinV-06} can be seen as the analogue of our above mentioned condition on the covariance function in Section \ref{sec:abstract}.
\section{Notation and main result}
Let $d \in \NN$, $(\Omega , \mathcal{A} , \PP)$ be a probability space and $V = \{V_x : \Omega \to \RR,\ x \in \ZZ^d \}$ a stationary Gaussian process with mean zero and covariance function $\gamma : \ZZ^d \to \RR$.
This implies that any finite combination $\sum_{k=1}^n \alpha_k V_{x_k}$, $\alpha_k \in \RR$, $x_k \in \ZZ^d$, is a Gaussian random variable, the two random vectors
\[
 (V_{x_1} , V_{x_2} , \ldots , V_{x_n}) \quad \text{and} \quad (V_{x_1+y} , V_{x_2+y} , \ldots , V_{x_n+y})
\]
have for any $y \in \ZZ^d$ the same (Gaussian) probability distribution, and for all $x,y \in \ZZ^d$ we have
\[
 \EE (V_x) = 0
 \quad \text{and} \quad
 \cov \left( V_x , V_y \right) = \gamma (x-y) = \gamma (y-x) .
\]
Here, $\EE$ denotes the expectation with respect to the probability measure $\PP$, and $\cov (V_x , V_y) \allowbreak = \allowbreak \EE (V_x V_y) - \EE (V_x)\EE (V_y)$ denotes the covariance of $V_x$ and $V_y$.
We assume that $0 < \gamma (0) < \infty$ and that there are $\alpha , D > 0$ such that
\[
 \lvert \gamma (x) \rvert \leq D \euler^{-\alpha \lvert x \rvert_1} .
\]
Note that the covariance function may have unbounded support and is allowed to change its sign arbitrarily. By Cauchy Schwarz inequality we have for all $x \in \ZZ^d$
\begin{equation} \label{eq:bounded}
 \lvert \gamma (x) \rvert = \cov (V_0 , V_x) \leq \sqrt{\var (V_0)  \var (V_y)} = \gamma (0) .
\end{equation}
\par
%
For $L > 0$ we introduce the notation $\Lambda_L = [-L,L]^d \cap \ZZ^d$, denote by $A_L$ an arbitrary self-adjoint operator in $\ell^2 (\Lambda_L)$, and for $\omega \in \Omega$ we denote by $V_{\omega , L}$ the multiplication operator on $\ell^2 (\Lambda_L)$ by the function $\Lambda_L \ni x \mapsto V_{\omega , L} (x) = V_x (\omega)$. For each $\omega \in \Omega$ and $\lambda > 0$ we introduce the finite volume Schr\"odinger operator
\[ 
 H_{\omega , L} = A_L + \lambda V_{\omega , L} 
\]
in $\ell^2 (\Lambda_L)$. Our main result is the following theorem.
\begin{theorem} \label{theorem:wegner}
There are constants $C_{\mathrm{W}} > 0$ and $I_0 = I_0 \in \NN_0^d$, both depending only on the covariance function $\gamma$, such that for any $L > 0$, any bounded interval $I \subset \RR$, and any $\lambda > 0$
\[
 \EE \left( \Tr \bigl(\chi_I (H_{\omega , L}) \bigr) \right) \leq
\frac{C_{\mathrm{W}}}{\lambda} \lvert I \rvert (2L+1)^{2d + \lvert I_0 \rvert_1} .
\]
\end{theorem}
The proof of Theorem~\ref{theorem:wegner} is divided into three steps. First we provide an abstract Wegner estimate in Theorem~\ref{theorem:abstract_wegner}. It states that if a certain transformation of the covariance function $\gamma$ is non-negative, see Ineq.~\eqref{eq:positivity}, then a Wegner estimate follows. 
In a second step we cite a result of \cite{LeonhardtPTV-15}, which allows us to verify Ineq.~\eqref{eq:positivity} for exponentially decaying covariance functions. In a last step we combine these two results to prove the Wegner estimate stated in Theorem~\ref{theorem:wegner}.
\section{An abstract Wegner estimate} \label{sec:abstract}
The following theorem may be understood, e.g., as a discrete variant of Theorem 1 in \cite{FischerHLM-97}. However, our assumption on the covariance function is weaker than the discrete analogue of \cite{FischerHLM-97}, in the sense that Ineq.~\eqref{eq:positivity} is required for $x \in \Lambda_L$ instead of $x \in \ZZ^d$. 
This observation is essential since for the class of exponentially decaying and sign-changing covariance functions, we are able to verify Ineq.~\eqref{eq:positivity} for all $x \in \Lambda_L$, but not for all $x \in \ZZ^d$. 
\par
Similar conditions as in Ineq.~\eqref{eq:positivity} were obtained before in proofs of Wegner estimates for the alloy-type model in the discrete and continuous setting, see, e.g., \cite{KostrykinV-06,Veselic-10,Veselic-10b,LeonhardtPTV-15}.
\begin{theorem} \label{theorem:abstract_wegner}
Assume there is $L_0 > 0$ such that for arbitrary $L \geq L_0$ and every $j \in \Lambda_L$ there is a compactly supported sequence $t_{j,L} \in \ell^1 (\ZZ^d)$ such that
\begin{equation} \label{eq:positivity}
\sum_{k \in \ZZ^d} t_{j,L} (k) \gamma (x-k) \geq \delta_j (x) \quad \text{for all} \quad x \in \Lambda_L .
\end{equation}
Let further $I = [E_1,E_2]$ be an arbitrary interval. Then for any $L \geq L_0$ we have
\begin{equation*}
\EE \{ \Tr \chi_I (H_{\omega , L})\} \leq \frac{1}{\sqrt{2\pi} \lambda} \lvert I \rvert \sum_{j \in \Lambda_L} \sqrt{ \sum_{k,l \in \ZZ^d} t_{j,L} (k) t_{j,L} (l) \gamma (k-l) } .
\end{equation*}
\end{theorem}
\begin{remark}
 One might wonder whether assumption \eqref{eq:positivity} for $j = 0$ implies assumption \eqref{eq:positivity} for all $j \not = 0$ by taking suitable translates. This is not the case, since assumption \eqref{eq:positivity} is required only for $x \in \Lambda_L$ instead of $x \in \ZZ^d$.
\end{remark}

For the proof of Theorem~\ref{theorem:abstract_wegner} we will use an estimate on averages of spectral projections of certain self-adjoint operators. More precisely, on a Hilbert space $\mathcal{H}$, let $H$ be self-adjoint, $U$ symmetric and $H$-bounded, $J$ bounded and non-negative with $J^2 \leq U$, $H (\zeta) = H + \zeta U$ for $\zeta \in \RR$, and $\chi_I (H(\zeta))$ the corresponding spectral projection onto an interval $I \subset \RR$. Then, for any $g \in L^\infty (\RR) \cap L^1 (\RR)$, $\psi \in \mathcal{H}$ with $\lVert \psi \rVert = 1$ and bounded interval $I \subset \RR$, we have
\begin{equation} \label{eq:average_proj}
\int_\RR \bigl\langle \psi , J \chi_I (H(\zeta)) J \psi \bigr\rangle g(\zeta) \drm \zeta \leq \lVert g \rVert_\infty \lvert I \rvert .
\end{equation}
For a proof of Ineq.~\eqref{eq:average_proj} we refer to \cite{CombesH-94} where compactly supported $g$ is considered. The non-compactly supported case was first treated in \cite{FischerHLM-97}, see also \cite[Lemma~5.3.2]{Veselic-08} for a detailed proof.
\begin{proof}[Proof of Theorem~\ref{theorem:abstract_wegner}]
In order to estimate the expectation of the trace
\[
\Tr \chi_I (H_{\omega , L}) = \sum_{j \in \Lambda_L} \lVert \chi_I(H_{\omega , L}) \delta_j \rVert^2
\]
we fix $L \geq L_0$ and $j \in \Lambda_L$, and use the notation $t = t_{j,L}$. 
Let $W : \ZZ^d \to \RR$ and $\kappa : \Omega \to \RR$ be given by 
\[
 W (x) = \frac{\lambda}{\sqrt{N}} \sum_{k \in \ZZ^d} \gamma (x-k) t (k)
 \quad\text{and}\quad
 \kappa = N^{-1/2} \sum_{k \in \ZZ^d} V_{k}  t (k)
\]
where $N$ denotes the normalization constant $N = \sum_{k,l \in \ZZ^d} t (k) t (l) \gamma (k-l)$. Note that $N > 0$, since $N$ is the variance of the linear combination $\sum_{k \in \ZZ^d} V_k t (k)$.
We consider the decomposition
\[
 H_{\omega , L} = A_L + \lambda V_{\omega , L} = B + \kappa (\omega) W
\quad \text{where} \quad
 B = A_L  + \lambda V_{\omega , L} - \kappa (\omega) W .
\]
By construction, $\kappa$ is normally distributed with mean zero and variance one.
Moreover, $\kappa$ is independent of the $\sigma$-algebra 
$
 \mathcal{F}:= \sigma \left(\lambda V_k - \kappa W (x) \colon x \in \ZZ^d \right)
$. This follows from the fact that for all $x \in \ZZ^d$ we have
 \[
  \cov (\kappa , \lambda V_x - \kappa W (x))
  =
 \lambda \EE (\kappa V_x) - W (x) = \frac{\lambda}{N^{1/2}} \sum_{k \in \ZZ^d} \gamma (x-k) t (k) - W (x) = 0 .
  \]
Hence, we obtain 
\begin{equation*}
 \EE \left(\lVert \chi_I(H_{\omega , L}) \delta_j \rVert^2 \right) 
  = 
 \EE\left( \EE \left(\lVert \chi_I(H_{\omega , L}) \delta_j \rVert^2 \mid \mathcal{F} \right) \right)
 = 
 \EE\left( 
 \int_\RR \langle \delta_j , \chi_I (B + \kappa W) \delta_j \rangle \frac{\mathrm{e}^{-\kappa^2 / 2}}{\sqrt{2\pi}} \drm \kappa
 \right) .
\end{equation*}
By assumption, $W$ satisfies $W (x) \geq \lambda N^{-1/2} \delta_j (x)$ for all $x \in \Lambda_L$. Hence, we can apply Ineq.~\eqref{eq:average_proj} with $H = B$, $\zeta = \kappa$, $U = W$, $J^2 = \lambda N^{-1/2} \delta_j$ and $g$ the standard Gaussian density to obtain
\[
 \EE \left(\lVert \chi_I(H_{\omega , L}) \delta_j \rVert^2 \right) 
 = \sqrt{\frac{N}{2\pi \lambda^2}} \EE\left( 
 \int_\RR \langle \delta_j , J \chi_I (B + \kappa W) J \delta_j \rangle \mathrm{e}^{-\kappa^2 / 2} \drm \kappa
 \right) 
 \leq \sqrt{\frac{N}{2\pi \lambda^2}}  \lvert I \rvert .
\]
The result follows by summing over $j \in \Lambda_L$.
\end{proof}

\section{Linear combinations of translated covariance functions}
In this subsection we cite a result of Leonhardt, Peyerimhoff, Tautenhahn and Veseli\'c \cite{LeonhardtPTV-15}. This will allow us to ensure the positivity condition~\eqref{eq:positivity} for arbitrary sign-changing and exponentially decaying covariance functions. Recall that $\lvert \gamma (x) \rvert \leq D \exp (-\alpha \lvert x \rvert_1)$ by assumption.
\par
In order to formulate the result of \cite{LeonhardtPTV-15} we introduce some notation. For $I
=(i_1,\dots,i_d) \allowbreak \in \ZZ^d$ and $z \in \CC^d$, we define
\[
z^I = z_1^{i_1} \cdot z_2^{i_2} \cdot \ldots \cdot z_d^{i_d} .
\]
For $I \in \NN_0^d$, we use the notation
\begin{align*}
  D_z^I = \frac{\partial^{i_1}}{\partial z_1^{i_1}} \cdot \frac{\partial^{i_2}}{\partial z_2^{i_2}} \cdot \ldots \cdot \frac{\partial^{i_d}}{{\partial z_d}^{i_d}} .
\end{align*}
We also introduce comparison symbols for a multi-index: If $I, J \in \NN_0^d$, we write $J \le I$ if we have $j_r \le i_r$ for all
$r=1,2,\ldots,d$, and we write $J < I$ if $J \le I$ and $\lvert J \rvert_1 <
\lvert I \rvert_1$. For $\delta \in (0, 1-\euler^{-\alpha})$ we consider the generating function $F : D_\delta \subset \CC^d \to \CC$, 
\begin{equation*} \label{eq:Fz}
D_\delta = \{ z \in \CC^d : \lvert z_1 - 1 \rvert < \delta , \ldots , \lvert z_d - 1 \rvert < \delta \}, \quad F(z) = \sum_{k \in \ZZ^d} \gamma(-k) z^k .
\end{equation*}
The function $F$ is a holomorphic function, see \cite{LeonhardtPTV-15} for details. Since $F$ is holomorphic and not identically zero, we have $(D_z^I F) (\mathbf{1}) \not = 0$ for at least one $I \in \NN_0^d$. Here $\mathbf{1} = (1,\ldots , 1)^{\mathrm T} \in \RR^d$. Hence, there is $I_0 \in \NN_0^d$ (not necessarily unique) and $c \not = 0$, such that
\begin{equation*} \label{eq:cF}
  (D_z^I F)({\mathbf 1}) = 
		\begin{cases} 
			c \neq 0, & \text{if $I = I_0$,} \\
			0,         & \text{if $I < I_0$.} 
		\end{cases}  
\end{equation*}
\begin{proposition}[{\cite[Proposition~4.2]{LeonhardtPTV-15}}] \label{prop2}
Let $\gamma : \ZZ^d \to \RR$, $c \not = 0$ and $I_0 \in \NN_0^d$ be as above. Let further $L > 0$ and define
\begin{equation*} \label{eq:RLrel} 
  R_L = \max \left\{ 2L + \frac{2}{\alpha} \ln \frac{2\, 3^d\, D}
  {\lvert c \rvert (1-\mathrm{e}^{-\alpha/2})}, \frac{8 (d+\lvert I_0 \rvert_1)^2}{\alpha^2} \right\}. 
\end{equation*}
Then we have for all $x \in \Lambda_L$
\[ 
  \frac{2}{c} \sum_{k \in \Lambda_{R_L}}  k^{I_0}\, \gamma(x-k) \ge 1.
\] 
\end{proposition}

\section{Proof of Theorem~\ref{theorem:wegner}}
Let $L_0 > 0$ be arbitrary. By Proposition~\ref{prop2}, Assumption~\eqref{eq:positivity} of Theorem~\ref{theorem:abstract_wegner} is satisfied with $t_{j,L} \in \ell^1 (\ZZ^d)$ given by 
\[
t_{j , L} (k) = \begin{cases}
                 2 k^{I_0} / c & \text{if $k \in \Lambda_{R_L}$}, \\
		0 & \text{else},
                \end{cases}
\]
for $L \geq L_0$ and $j \in \Lambda_{L}$. For the $\ell^1$-norm of $t_{j,L}$ we obtain
\begin{align*}
  \lVert t_{j,L} \rVert_{\ell^1 (\ZZ^d)} & \leq \frac{2}{\lvert c \rvert} \sum_{k \in \Lambda_{R_L}} \lvert k^{I_0} \rvert \leq 
\frac{2}{\lvert c \rvert}  (2R_L + 1)^d R_L^{\lvert I_0 \rvert_1} .
\end{align*}
By Proposition~\ref{prop2}, $R_L = \max \{2L + D' , D''\} < 2L+D'+D''$ with $D'$ and $D''$ depending only on the covariance function. Hence there is a constant $C_{\mathrm W}' > 0$ depending only on the covariance function such that
\begin{equation} \label{eq:volume}
 \sum_{j \in \Lambda_{L}} \lVert t_{j,L} \rVert_{\ell^1 (\ZZ^d)} \leq
C_{\mathrm W}' (2L+1)^{2d + \lvert I_0 \rvert_1}.
\end{equation}
The result now follows from Theorem~\ref{theorem:abstract_wegner}, Ineq.~\eqref{eq:bounded}, and Ineq.~\eqref{eq:volume}. \qed
\section{Regularity properties for stationary Gaussian processes}
\label{sec:conditional}
In this section we show that the abstract regularity conditions from \cite{DreifusK-91,AizenmanM-93,AizenmanG-98,Hundertmark-00,AizenmanSFH-01,Hundertmark-08} are in general not satisfied for discrete Gaussian processes. This shows that our result is not covered by the just mentioned references.
Let $Z_0^\perp = \times_{k \in \ZZ^d \setminus \{0\}} \RR$ and $\mathcal{Z}_0^\perp = \otimes_{k \in \ZZ^d \setminus \{0\}} \mathcal{B} (\RR)$. We introduce the random variable
\[
 V_0^\perp : (\Omega, \cA) \to (Z_0^\perp , \mathcal{Z}_0^\perp), \quad V_0^\perp (\omega) = (V_k (\omega))_{k \in \ZZ^d \setminus \{ 0 \}} .
\]
We denote by $\PP_{0^\perp} : Z_0^\perp \to [0,1]$ the distribution of $V_0^\perp$ with respect to $\PP$, i.e.\ $\PP_{0^\perp} (B) := \PP (\{\omega \in \Omega \colon V_0^\perp (\omega) \in B\})$.
For $a \in \RR$ and $\epsilon > 0$ we set
\[
  Y^{\epsilon,a} := \PP \bigl( V_0 \in [a,a+\epsilon] \mid V_0^\perp \bigr) :=  \EE \bigl( \Eins_{\{ V_0 \in [a,a+\epsilon] \}} \mid V_0^\perp \bigr) .
\]
For convenience, for each $a \in \RR$ and $\epsilon >0$ we fix one version $Y^{\epsilon,a}$ of the conditional expectation.
Since $Y^{\epsilon,a}$ is $\sigma (V_0^\perp)$-measurable, the factorization lemma tells us that (for each $a$ and $\epsilon$) there is a measurable function $g^{\epsilon , a} : (Z_0^\perp , \mathcal{Z}_0^\perp) \to (\RR , \cB (\RR) )$ 
such that $Y^{\epsilon,a} = g^{\epsilon , a} \circ V_0^\perp$, i.e.\ for almost all $\omega \in \Omega$ we have
\begin{equation*} \label{eq:factor}
 Y^{\epsilon , a} (\omega) = g^{\epsilon , a}(V_0^\perp (\omega)) .
\end{equation*}
For $\epsilon > 0$ we define conditional concentration function
\[
 S (\epsilon) := \sup_{a \in \RR} \esssup_{v \in Z_0^\perp}   g^{\epsilon,a}(v) .
\]
Here, the essential supremum refers to the measure $\PP_{V_0^\perp}$, that is, 
\[
 \esssup_{v \in Z_0^\perp}   g^{\epsilon,a}(v) = \inf \Bigl\{b \in \RR \colon \PP_{V_0^\perp} \bigl(\{v \in Z_0^\perp \colon g^{\epsilon , a} (v) > b\}\bigr) = 0  \Bigr\} .
\]
We formulate exemplary the regularity condition from \cite{AizenmanSFH-01}.
\begin{definition}
 The collection $V_k$, $k \in \ZZ^d$, is said to be conditional $\tau$-H\"older continuous for $\tau \in (0,1]$, if there is a constant $C$ such that for all $\epsilon > 0$
 \[
  S (\epsilon) \leq C \epsilon^\tau .
 \]
\end{definition}
The next theorem shows that the Gaussian process $V = \{V_x : \Omega \to \RR,\ x \in \ZZ^d \}$ is not $\tau$-H\"older continuous, if $\gamma (0) = 2$, $\gamma (-1) = \gamma (1) = 1$ and $\gamma (k) = 0$ for $k \in \ZZ \setminus \{-1,0,1\}$.
\begin{theorem}\label{thm:regularity} Let $d = 1$, $\gamma (0) = 2$, $\gamma (-1) = \gamma (1) = 1$ and $\gamma (k) = 0$ for $k \in \ZZ \setminus \{-1,0,1\}$. Then for any $\epsilon > 0$ we have
\[
 S (\epsilon) = 1 .
\]
\end{theorem}

\begin{proof}
Let $l \in \NN$, $V_+^l := (V_k)_{k=1}^l$, $V_-^l := (V_{-l+k-1})_{k=1}^l$ and $V^l = (V_l^- , V_l^+) \in \RR^{2l}$. First we note that the distribution of $V_0$ conditioned on $V^l = v \in \RR^{2l}$ is again Gaussian with variance
\[
 \gamma_l = \gamma (0) - \cov (V_0 , V^l) \cov (V^l , V^l)^{-1} \cov (V^l , V_0 ) 
\]
and mean
\[
 m_l = \cov (V_0 , V^l) \cov (V^l , V^l)^{-1} v ,
\]
see e.g.\ Proposition~3.6 in \cite{Port-94}. By assumption we have
\[
 \cov (V^l , V^l) 
 = 
 \begin{pmatrix}
  \Gamma_l & 0 \\
  0   & \Gamma_l
 \end{pmatrix} ,
\quad \text{where} \quad
 \Gamma_l
 = 
 \begin{pmatrix}
2 	& 1 	& 		& 		\\
1 	& 2 	& \ddots 	& 		\\
 		& \ddots 	& \ddots 	& 1	\\
  		& 		& 1 	& 2
\end{pmatrix} \in \RR^{l \times l}. 
\]
For its inverse we have by Cramer's rule
\begin{equation*} \label{eq:inverseelement}
 \Gamma_l^{-1} (1,1) = \Gamma_l^{-1} (l,l) = \frac{l}{l+1} .
\end{equation*}
Since $\cov (V_0 , V^l) =  (0,\ldots,0,1,1,0,\ldots,0) \in \RR^{2l}$ and $\gamma (0) = 2$ we find that
\begin{equation} \label{eq:variance}
 \gamma_l = 2\left(1 - \frac{l}{l+1}\right) .
\end{equation}
Let $\epsilon > 0$, $a \in \RR$, $b := \esssup_{V_0^\perp} g^{\epsilon, a} (V_0^\perp) \in [0,1]$, and $\delta > 0$.
By definition of the conditional expectation we have for all $B \in \sigma (V_0^\perp)$ that
\begin{equation} \label{eq:definition}
 \EE \bigl( \Eins_B \Eins_{\{V_0 \in [a,a+\epsilon]\}} \bigr) = \EE \bigl( \Eins_B Y_0^{\epsilon,a} \bigr) .
\end{equation} 
We choose
\[
 B = B_{l,\delta} = \bigl\{ \omega \in \Omega \colon V_k (\omega) \in [-\delta , \delta], \ k \in \{-l , \ldots , l\} \setminus \{0\} \bigr\} .
\] 
For the right hand side of Eq.~\eqref{eq:definition} we have
\[
\EE \bigl( \Eins_{B_{l,\delta}} Y_0^{\epsilon,a} \bigr) \leq b \PP (B_{l,\delta}) . 
\]
The left hand side of Eq.~\eqref{eq:definition} equals
\[
 \PP ( B_{l,\delta} \cap \{V_0 \in [a,a+\epsilon]\} ) = 
  \PP ( \{V_0 \in [a,a+\epsilon]\} \mid B_{l,\delta} )\cdot \PP(B_{l,\delta}) .
\]
Hence, we have
\[
\PP ( \{V_0 \in [a,a+\epsilon]\} \mid B_{l,\delta} )
\leq 
b .
\]
By the definition of the conditional expectation we find that
\[
  \PP ( \{V_0 \in [a,a+\epsilon]\} \mid B_{l,\delta} ) \to \PP ( \{V_0 \in [a,a+\epsilon]\} \mid V^l = 0 ) = \mathcal{N}_{0 ,\gamma_l} ([a,a+\epsilon])
\]
as $\delta \to 0$. If $a = -\epsilon / 2$, then $\mathcal{N}_{0 ,\gamma_l} ([a,a+\epsilon]) \to 1$ as $l \to \infty$ by \eqref{eq:variance}. Hence we find $b = 1$.
\end{proof}

\subsection*{Acknowledgment}
The author gratefully acknowledges stimulating discussions with Ivan Veseli\'c and Christoph Schumacher.

\end{document}